\documentclass[a4paper,12pt]{article}

\usepackage[T1]{fontenc}
\usepackage[utf8]{inputenc}
\usepackage{microtype}
\usepackage{amsmath,amssymb,amsthm,mathtools}

\usepackage{authblk}
\usepackage{graphicx}
\usepackage{tikz}
\usetikzlibrary{arrows.meta}
\usepackage{float}

\usepackage{enumitem}
\usepackage[font=small,labelfont=bf,labelsep=colon]{caption}
\usepackage{geometry}
\usepackage{hyperref}
\hypersetup{hidelinks}

\usepackage[svgnames, table]{xcolor}
\usepackage[normalem]{ulem}

\newtheorem{theorem}{Theorem}[section]
\newtheorem{lemma}[theorem]{Lemma}
\newtheorem{definition}[theorem]{Definition}
\newtheorem{corollary}[theorem]{Corollary}

\theoremstyle{remark}
\newtheorem{remark}[theorem]{Remark}
\theoremstyle{plain}

\newcommand{\gap}{\text{\texttt{-}}}
\newcommand{\Ptime}{\ensuremath{\mathsf{P}}}
\newcommand{\NP}{\ensuremath{\mathsf{NP}}}

\newcommand{\MAXSNP}{\textsf{MAXSNP}}

\newcommand{\MSAS}{\ensuremath{\mathsf{MSA\text{-}S\text{-}DEC}}}

\newcommand{\MSAOPTlambda}{\ensuremath{\mathsf{MSA\text{-}S\text{-}OPT}_{\lambda}}}

\title{Structure-Informed Multiple Sequence Alignment: \\A Formal Model and Hardness Results}

\author[1]{Kanazawa~Yoshiki}
\author[2]{Naphan~Benchasattabuse}
\author[2]{Michal~Hajdušek}
\author[3]{Rodney~Van~Meter}

\affil[1]{Graduate School of Media and Governance, Keio University, Fujisawa, Kanagawa, Japan}
\affil[2]{Graduate School of Media Design, Keio University, Yokohama, Kanagawa, Japan}
\affil[3]{Faculty of Environment and Information Studies, Keio University, Fujisawa, Kanagawa, Japan}
\date{\vspace{-1em}\small
\texttt{ky01752221@keio.jp}, \texttt{\{whit3z, michal, rdv\}@sfc.wide.ad.jp}
}

\begin{document}
\maketitle

\begin{abstract}
We formulate a structure-informed multiple sequence alignment problem, denoted \textsf{MSA-S}.
The model abstracts biological sequences as strings and structural information as designated position-pairs.
It augments a fixed pairwise string score—defined by a fixed non-gap symbol-pair scoring rule and fixed affine gap penalties—with a binary overlap score on designated position-pairs, which can be interpreted as a contact-map overlap score in structural applications. This yields a fixed-score, integer-valued optimization model suitable for complexity-theoretic analysis.

Under this formulation, we show that the decision problem \textsf{MSA-S-DEC} is NP-complete for a broad class of fixed pairwise string scoring schemes. We also show that NP-hardness persists even under the restriction that every designated position-pair set is nonempty and the pair-overlap threshold is strictly positive.
For the associated scalarized optimization problem \MSAOPTlambda\ with any fixed rational constant \(\lambda \ge 1\), we further show that, under the canonical unit scheme for the non-gap symbol-pair scoring rule, \MSAOPTlambda\ admits no polynomial-time approximation scheme (PTAS) even for two input strings
\((k=2)\), unless \(\mathsf{P}=\mathsf{NP}\).

These results establish a formal complexity-theoretic baseline for structure-informed multiple sequence alignment.
\end{abstract}

\section{Introduction}
Multiple Sequence Alignment (MSA) is a foundational technique in bioinformatics for comparing related biological sequences, such as DNA, RNA, or protein sequences. By identifying patterns of conservation and variation among them, MSA provides information about functional constraints and evolutionary relationships, and can also inform structural interpretation~\cite{Chatzou2016,Chao2022}.
It is therefore widely used as a preliminary step for downstream analyses such as phylogenetic inference, protein family and domain analysis, and functional annotation~\cite{Warnow2021}.

Computationally, an MSA represents these comparisons by arranging the input sequences into a common column representation, usually by inserting gaps while preserving the order of symbols in each sequence. Symbols placed in the same column are interpreted as hypothesized positional correspondences. The quality of an alignment is then evaluated by a scoring scheme, for example by rewarding
compatible symbols and penalizing gaps, often with gap-opening and gap-extension penalties~\cite{Durbin1998,Gotoh1982}.

Although optimal pairwise alignment can be computed efficiently by dynamic programming~\cite{Needleman1970,Smith1981}, exact dynamic programming formulations of MSA scale poorly with the number of sequences. More fundamentally, standard MSA formulations based on the sum-of-pairs (SP) score are known to be computationally intractable: Wang and Jiang proved NP-completeness for the SP-score decision version of MSA, and subsequent work established NP-hardness for broader scoring classes and common variants of multiple alignment~\cite{Wang1994,Just2001,Elias2006}. Consequently, practical MSA methods rely on heuristic strategies and make different trade-offs among accuracy, speed, and applicability to different classes of input data~\cite{Edgar2006,Chatzou2016,Chao2022,Prousalis2025}.

Within this broader landscape, a line of work has explored auxiliary information beyond sequence similarity. In protein sequence alignment, structural information is a natural source of such auxiliary information: available structural information can provide constraints on positional compatibility, and existing structure-aware alignment methods incorporate such information either directly or in combination with sequence-based scoring~\cite{OSullivan2004,Pei2008,Rozewicki2019}.

A comparative benchmark by Carpentier and Chomilier evaluated sequence-based, structure-based, and sequence--structure-based alignment programs against reference alignments across multiple databases, and suggested that structural information can be beneficial but that its contribution is not uniform across methods and benchmark settings~\cite{Carpentier2019}. This indicates that the effect of structural information depends not only on the availability of an additional data source, but also on how that information is represented and combined with sequence-based scoring.

These observations motivate studying a fixed-score combinatorial optimization model in which the sequence-based alignment score and the auxiliary pair-based information are represented explicitly and can be analyzed separately.

In this paper, we formulate a structure-informed multiple sequence alignment problem, denoted \textsf{MSA-S}. To make the model accessible as a combinatorial problem, we abstract biological sequences as strings over a fixed finite alphabet, and abstract structural information as designated position-pairs in each string. The model is intentionally minimal: it augments a fixed pairwise string score, specified by a fixed non-gap symbol-pair scoring rule and fixed affine gap penalties, with a binary overlap score on designated position-pairs. In structural applications, these designated pairs may be viewed as contact-map pairs, obtained for example by thresholding distances between spatially close positions~\cite{Goldman1999}. This yields a fixed-score, integer-valued optimization model that is simple enough for worst-case complexity analysis while still making the sequence and structure contributions explicit.

We show that the associated decision problem is NP-complete for a broad class of fixed pairwise string scoring schemes, and that NP-hardness persists even under the restriction that every designated position-pair set is nonempty and the pair-overlap threshold is strictly positive. We further show that, under the canonical unit scheme for the non-gap symbol-pair scoring rule, with the fixed affine gap penalties retained, the corresponding scalarized optimization problem admits no polynomial-time approximation scheme (PTAS) even for two input strings \((k=2)\), unless \(\mathsf{P}=\mathsf{NP}\).

Together, these results provide a formal complexity-theoretic baseline for structure-informed multiple sequence alignment. 
They show that even a deliberately simplified fixed-score model is computationally intractable in its decision form, and that its scalarized optimization variant does not admit a PTAS under the stated assumptions.

\section{Model and problem definitions}\label{sec:model}
We formulate a combinatorial optimization problem on strings with auxiliary pair information. The goal of this section is to define the problem in a purely string-based way, without assuming prior familiarity with biological sequence alignment.

\subsection{Informal description}
Before giving the formal definition, we first describe what an alignment represents.
Given \(k\) input strings, an alignment is obtained by inserting gap symbols into the strings. Gap insertions create empty positions but do not change the left-to-right order of the original symbols in any string. After gaps have been inserted, all rows have the same length. Symbols appearing in the same column are then interpreted as being aligned with one another.

For example, if two strings are written as
\[
S_1 = \texttt{ACG},
\qquad
S_2 = \texttt{AG},
\]
then one possible alignment is
\[
\begin{array}{ccc}
\texttt{A} & \texttt{C} & \texttt{G}\\
\texttt{A} & \gap     & \texttt{G}
\end{array}
\]
where a gap has been inserted between the two symbols of \(S_2\). Thus, the first symbols of the two strings are aligned, the third symbol of \(S_1\) is aligned with the second symbol of \(S_2\), and the second symbol of \(S_1\) is placed opposite the inserted gap.

We next explain how structural information is represented in this model.
In structural applications, three-dimensional information is first converted into binary contact information on pairs of positions. For each sequence \(S_i\), assume that a known or predicted three-dimensional structure is available. A pair of positions is included in the contact set \(C_i\) when the corresponding residues are sufficiently close in the structure, such as when their spatial distance is below a fixed threshold. Thus, \(C_i\) records which pairs of positions are treated as contacts.

Given an alignment, each original position is assigned to an alignment column.
Thus, a contact pair in \(C_i\) is mapped to the two alignment columns occupied by its two positions. The pair-overlap score rewards cases in which contact pairs from different sequences are mapped to the same two alignment columns. 
For example, if a contact pair in one sequence is placed in columns 1 and 3, and a contact pair in another sequence is also placed in columns 1 and 3, then these two contact pairs contribute one unit to the pair-overlap score. In this sense, the score measures whether pairwise structural relationships are preserved by the alignment.

For the formal model studied in this paper, however, the construction of the sets \(C_i\) is not part of the problem. We assume that the sets \(C_i\) are already given. They may come from structural data, for example by thresholding distances between positions, but in the combinatorial formulation they are arbitrary designated sets of position-pairs.

\subsection{Basic objects and alignments}
We first specify the objects to be aligned. The model starts with \(k\) input strings, where \(k\) is the number of strings. Each input string is represented as a finite sequence of symbols. The set of symbols that may appear in the input strings is denoted by \(\Sigma\). Throughout the paper, \(\Sigma\) is treated as a fixed finite alphabet, and we assume \(|\Sigma|\ge 2\). We write \(S_i\) for the \(i\)-th input string, and \(L_i\) for its original length. Thus, the positions of \(S_i\) are indexed by \(1,\ldots,L_i\), and each position contains a symbol from \(\Sigma\).

An alignment may insert gap symbols into the strings in order to place their symbols into common columns. We denote the gap symbol by \(\gap\). The gap symbol is not an ordinary input symbol, so we assume \(\gap\notin\Sigma\). In other words, the original strings contain only symbols from \(\Sigma\), while \(\gap\) is introduced only when forming an alignment.

An alignment of \(S_1,\ldots,S_k\) is represented as a \(k\times L\) matrix
\[
A\in(\Sigma\cup\{\gap\})^{k\times L}
\]
for some integer \(L\), where the rows correspond to the input strings and the columns are alignment columns. The value of \(L\) is the alignment length; it may be larger than the length of an individual input string because gaps may be inserted. Each row of \(A\) is obtained from the corresponding input string by inserting gap symbols, so deleting all gaps from row \(i\) recovers \(S_i\). We exclude all-gap columns, since they do not affect any positional correspondence or score.

\subsection{Designated position-pair sets}
For each input string \(S_i\), the input also includes a set
\[
C_i \subseteq \{(p,q): 1\le p<q\le L_i\}
\]
of designated position-pairs. An element \((p,q)\in C_i\) means that the pair of positions \(p\) and \(q\) in \(S_i\) is used in the pair-overlap score defined below. Position-pairs not in \(C_i\) are ignored by that score.

Such a set may be derived from structural information, for example by selecting pairs that satisfy a distance-threshold condition. In the formal model studied here, however, the construction of \(C_i\) is not part of the problem; \(C_i\) is simply given as part of the input.

\subsection{Alignment-induced positional correspondence}
After an alignment \(A\) has been fixed, it determines which original positions are placed in the same alignment column. For a position \(p\) in \(S_i\) and a position \(q\) in \(S_j\), we write
\[
(i,p)\leftrightarrow_A (j,q)
\]
if, after forming the alignment \(A\), the position \(p\) of \(S_i\) and the position \(q\) of \(S_j\) occupy the same alignment column. This notation refers only to positional correspondence induced by the alignment; it does not require the two symbols to be identical.

\subsection{Two score components}
The model assigns two scores to an alignment.
The first score, \(f_{\mathrm{seq}}\), evaluates the alignment of the symbols themselves: symbols placed in the same column should be compatible, and inserted gaps should be penalized. 
The second score, \(f_{\mathrm{str}}\), evaluates how the designated position-pairs are aligned across strings: designated pairs from different strings are rewarded when they are placed on the same two alignment columns.

Thus, each alignment \(A\) is associated with the score pair
\[
\bigl(f_{\mathrm{seq}}(A), f_{\mathrm{str}}(A)\bigr).
\]
The two scores are defined separately below.

\subsection{Pair-overlap score}\label{sec:str}
We now define \(f_{\mathrm{str}}\), using the designated position-pair sets and the alignment-induced correspondence introduced above. The pair-overlap score counts designated position-pairs from different strings that are placed on the same two alignment columns.

Consider \((p,q)\in C_u\) and \((p',q')\in C_v\). These two designated pairs are counted if they are placed on the same two alignment columns, namely if
\[
(u,p)\leftrightarrow_A (v,p')
\quad\text{and}\quad
(u,q)\leftrightarrow_A (v,q').
\]

Formally, we define
\[
f_{\mathrm{str}}(A)=
\sum_{1\le u<v\le k}
\sum_{(p,q)\in C_u}
\sum_{(p',q')\in C_v}
\mathbf{1}\!\left[
(u,p)\leftrightarrow_A (v,p')
\ \wedge\
(u,q)\leftrightarrow_A (v,q')
\right],
\]
to be maximized.
Here, \(\mathbf{1}[\mathcal{E}]\) equals \(1\) if the condition \(\mathcal{E}\)
holds and \(0\) otherwise.

\subsection{Pairwise string score}\label{sec:seq}
We now define \(f_{\mathrm{seq}}\), the sequence-based score. This score has two parts: a score for pairs of non-gap symbols placed in the same alignment column, and a penalty for gaps introduced by the alignment.

We first specify how to score pairs of non-gap symbols. Let
\[
s:\Sigma\times\Sigma\to\mathbb{Z}
\]
be a fixed non-gap symbol-pair scoring rule. For symbols \(x,y\in\Sigma\), the integer \(s(x,y)\) is the score assigned when \(x\) and \(y\) are placed in the same alignment column. In biological sequence-alignment terminology, such a rule is often called a substitution score.

We next specify how to score gaps. Gap penalties distinguish between starting a new gap block and extending it. For each pair of input strings \(S_u\) and \(S_v\), we consider the two rows of the alignment matrix \(A\) corresponding to these strings, ignoring columns in which both rows contain gaps. Gap runs are consecutive blocks of gaps in one row, opposite symbols from \(\Sigma\) in the other row, each taken as long as possible.
We denote by \(\mathcal{R}_{uv}(A)\) the set of such gap runs for the row pair \((u,v)\).

Each gap run \(r\in\mathcal{R}_{uv}(A)\) has length \(\ell_r\) and receives the affine penalty
\[
g_o+g_e\ell_r.
\]
Here \(g_o\le 0\) is the gap-opening penalty and \(g_e\le 0\) is the gap-extension penalty. Since \(f_{\mathrm{seq}}\) is maximized, gap penalties are represented as non-positive score contributions. Both parameters are fixed constants in the scoring scheme.

We now combine the non-gap symbol scores and the gap-run penalties to define \(f_{\mathrm{seq}}\). To write the formula, let \(a_{ic}\) denote the entry in row \(i\) and column \(c\) of the alignment matrix \(A\). Thus, \(a_{ic}\) is either a symbol from \(\Sigma\) or the gap symbol \(\gap\).

We define
\[
f_{\mathrm{seq}}(A)
=
\sum_{1\le u<v\le k}
\left(
\sum_{c=1}^{L}
\mathbf{1}\!\left[
a_{uc}\in\Sigma \ \wedge\ a_{vc}\in\Sigma
\right]
s(a_{uc},a_{vc})
+
\sum_{r\in\mathcal{R}_{uv}(A)}
\left(g_o+g_e\ell_r\right)
\right),
\]
to be maximized.
In the first inner sum, the factor 
\(\mathbf{1}[a_{uc}\in\Sigma \wedge a_{vc}\in\Sigma]\) ensures that \(s(a_{uc},a_{vc})\) is added only when both entries are non-gap symbols.
The second inner sum adds the affine penalties for all gap runs between the two rows.

\subsection{Decision and optimization problems}\label{sec:dec}
Having defined the two score components, we now define the associated decision and optimization problems. The decision version asks whether there exists an alignment whose two score components both meet prescribed thresholds. The optimization version combines the two scores into a single scalarized objective.

\begin{definition}[\MSAS]
Given strings \(S_1,\ldots,S_k\), designated position-pair sets \(C_1,\ldots,C_k\), and integer score thresholds
\(\tau_{\mathrm{seq}}\) and \(\tau_{\mathrm{str}}\), decide whether there exists
an alignment \(A\) of the given strings such that
\[
f_{\mathrm{seq}}(A)\ge\tau_{\mathrm{seq}}
\qquad\text{and}\qquad
f_{\mathrm{str}}(A)\ge\tau_{\mathrm{str}}.
\]
\end{definition}

\begin{definition}[\MSAOPTlambda]
For a fixed rational constant \(\lambda\ge 1\), given strings
\(S_1,\ldots,S_k\) and designated position-pair sets \(C_1,\ldots,C_k\), maximize
\[
f_{\mathrm{seq}}(A)+\lambda f_{\mathrm{str}}(A)
\]
over all alignments \(A\) of the given strings.
\end{definition}

\subsection{Instance encoding and NP-membership}\label{sec:norm}
The scoring scheme is fixed throughout the paper. The alphabet \(\Sigma\), the non-gap symbol-pair scoring rule \(s\), and the gap parameters \(g_o,g_e\) are treated as absolute constants and are not part of the input. For \(\MSAOPTlambda\), the parameter \(\lambda\ge 1\) is also treated as a fixed rational constant.

Unless explicitly stated otherwise, the number of strings \(k\) is part of the input.
An instance of the decision problem \(\MSAS\) consists of
\[
(S_1,\ldots,S_k;\ C_1,\ldots,C_k;\ \tau_{\mathrm{seq}},\tau_{\mathrm{str}}),
\]
where \(S_1,\ldots,S_k\) are strings, each
\[
C_i\subseteq\{(p,q):1\le p<q\le L_i\}
\]
is a designated position-pair set represented by listing its pairs of indices,
and \(\tau_{\mathrm{seq}},\tau_{\mathrm{str}}\) are integer score thresholds.

An instance of \(\MSAOPTlambda\) consists of
\[
(S_1,\ldots,S_k;\ C_1,\ldots,C_k),
\]
with the same convention for the designated position-pair sets. The scalarizing parameter \(\lambda\) is fixed in advance and is not part of the input.

\paragraph{NP-membership.}
The decision problem \(\MSAS\) belongs to \(\mathsf{NP}\). Indeed, a certificate is an alignment \(A\) of the given strings. Since all-gap columns are excluded, every alignment column contains at least one symbol from some input string, and hence the alignment length satisfies
\[
L\le \sum_{i=1}^k L_i.
\]
Thus \(A\) has polynomial size.

Given \(A\), the score \(f_{\mathrm{seq}}(A)\) can be evaluated in polynomial time by scanning the alignment columns and the gap runs for each pair of rows.
The score \(f_{\mathrm{str}}(A)\) can also be evaluated in polynomial time by checking, for the listed pairs in the sets \(C_i\), whether the corresponding positions are placed in the same alignment columns. Therefore the two threshold conditions in \(\MSAS\) can be verified in polynomial time. Hence \(\MSAS\in\mathsf{NP}\).

\section{NP-completeness by restriction}\label{sec:npc}
This section proves NP-completeness of \MSAS\ by a restriction transfer from the sequence-only alignment problem \textsf{Seq-MSA}. The source problem uses only the sequence score \(f_{\mathrm{seq}}\). The target problem is the decision problem \MSAS, whose instances additionally include designated position-pair sets and the pair-overlap score \(f_{\mathrm{str}}\).

The main reduction is by restriction. Given a \textsf{Seq-MSA} instance, we construct an \MSAS\ instance on the same strings by setting all designated position-pair sets to be empty and setting the pair-overlap score threshold \(\tau_{\mathrm{str}}\) to zero.
Then \(f_{\mathrm{str}}\) is identically zero for every alignment, so feasibility is exactly the same as in the original sequence-only instance.

We assume that the fixed scoring scheme defining \(f_{\mathrm{seq}}\) in Section~\ref{sec:model} is chosen from a class covered by Just's NP-hardness result for multiple alignment with sum-of-pairs (SP) score~\cite{Just2001}.

\begin{definition}[\textsf{Seq-MSA}]
An instance of \textsf{Seq-MSA} consists of
\[
(S_1,\ldots,S_k;\ \tau_{\mathrm{seq}}),
\]
where \(S_1,\ldots,S_k\) are strings and \(\tau_{\mathrm{seq}}\) is an integer score threshold. The problem asks whether there exists an alignment \(A\) of the given strings such that
\[
f_{\mathrm{seq}}(A)\ge \tau_{\mathrm{seq}}.
\]
\end{definition}

\begin{theorem}[Restriction transfer]\label{thm:msas-npcomplete}
Under the fixed scoring scheme of Section~\ref{sec:model}, if \textsf{Seq-MSA} is NP-hard, then \MSAS\ is NP-complete.
\end{theorem}

\begin{proof}
By the NP-membership argument in Section~\ref{sec:norm}, \(\MSAS\in\NP\).

For NP-hardness, we reduce from \textsf{Seq-MSA}. Given an instance
\[
(S_1,\ldots,S_k;\ \tau_{\mathrm{seq}})
\]
of \textsf{Seq-MSA}, construct the \MSAS\ instance
\[
(S_1,\ldots,S_k;\ \varnothing,\ldots,\varnothing;\ \tau_{\mathrm{seq}},0).
\]
That is, the strings are unchanged, every designated position-pair set is empty, and the two score thresholds are \(\tau_{\mathrm{seq}}\) and \(0\).

Since all designated position-pair sets are empty, there are no pair-overlap contributions. Hence, for every alignment \(A\),
\[
f_{\mathrm{str}}(A)=0.
\]
Therefore the second threshold condition
\[
f_{\mathrm{str}}(A)\ge 0
\]
is always satisfied, while the first threshold condition remains exactly
\[
f_{\mathrm{seq}}(A)\ge \tau_{\mathrm{seq}}.
\]
Thus the constructed \MSAS\ instance is feasible if and only if the original \textsf{Seq-MSA} instance is feasible.

The construction is polynomial-time computable and does not change the alignment space or the sequence score \(f_{\mathrm{seq}}\). Therefore \MSAS\ is NP-hard.
Together with \(\MSAS\in\NP\), this proves NP-completeness.
\end{proof}

\begin{corollary}\label{cor:just}
Suppose the fixed scoring scheme of Section~\ref{sec:model} belongs to the broad class covered by Just's Theorem~1(a)~\cite{Just2001}. 
Then \textsf{Seq-MSA} is NP-hard, and hence \MSAS\ is NP-complete.
\end{corollary}

\begin{proof}
Just proves NP-hardness of multiple alignment with sum-of-pairs (SP) score for a broad class of scoring matrices~\cite[Theorem~1(a)]{Just2001}, and notes that the result is unaffected by adding gap-opening penalties~\cite{Just2001}.

Just's result is stated in a minimization convention. Since minimizing a score is equivalent to maximizing its negation, with the threshold changed accordingly, the same hardness result applies to the maximization convention used here.

Since alignment length is polynomially bounded and all scoring parameters are fixed constants, all sequence scores are integer-valued and polynomially bounded in absolute value. Hence a polynomial-time algorithm for \textsf{Seq-MSA} would allow the optimum sequence score to be found by binary search over the feasible score range.
Thus \textsf{Seq-MSA} is NP-hard, and the claim follows from Theorem~\ref{thm:msas-npcomplete}.
\end{proof}

\begin{remark}[Nonempty pair-set strengthening]
The restriction-transfer proof above uses empty designated position-pair sets.
We next show that this restriction is not essential: NP-hardness persists even when every \(C_i\) is nonempty and \(\tau_{\mathrm{str}}>0\).
\end{remark}

\begin{lemma}[Nonempty designated-pair restriction does not remove NP-hardness]
\label{lem:nonempty-pairs}
Under the same fixed scoring scheme as above, \MSAS\ remains NP-hard even when each \(C_i\neq\varnothing\) and \(\tau_{\mathrm{str}}>0\).
\end{lemma}

\begin{proof}
We reduce from \textsf{Seq-MSA} instances with \(k\ge 2\), which is sufficient for NP-hardness. Let the input instance be
\[
(S_1,\ldots,S_k;\ \tau_{\mathrm{seq}}).
\]
For each \(i\), let \(L_i\) be the length of \(S_i\).
Fix any symbol \(x\in\Sigma\). We append the two-symbol suffix \(xx\) to each string and define
\[
S'_i := S_i\,\Vert\,xx \qquad (i=1,\ldots,k),
\]
where \(\Vert\) denotes concatenation. Thus the last two positions of \(S'_i\) are \(L_i+1\) and \(L_i+2\), and both contain the symbol \(x\).
We designate exactly this appended pair by setting
\[
C_i:=\{(L_i+1,L_i+2)\}.
\]
Here \(\binom{k}{2}\) is the number of pairs of distinct strings. We set
\[
\tau_{\mathrm{str}}:=\binom{k}{2},
\qquad
\tau'_{\mathrm{seq}}
:=
\tau_{\mathrm{seq}}+2\binom{k}{2}s(x,x).
\]
The term \(2\binom{k}{2}s(x,x)\) is the total sequence-score contribution of the two appended all-\(x\) columns: each such column contributes \(s(x,x)\) for each of the \(\binom{k}{2}\) pairs of rows.

(\emph{Completeness})
Suppose the original \textsf{Seq-MSA} instance is feasible, and let \(A\) be an alignment with
\[
f_{\mathrm{seq}}(A)\ge\tau_{\mathrm{seq}}.
\]
Append two columns to \(A\), each containing \(x\) in every row. The resulting alignment \(A'\) aligns the two appended positions of every string in the same two columns. Therefore
\[
f_{\mathrm{seq}}(A')
=
f_{\mathrm{seq}}(A)+2\binom{k}{2}s(x,x)
\ge
\tau'_{\mathrm{seq}}.
\]
Moreover, for every pair of strings \(u<v\), the unique designated pair in \(C_u\) and the unique designated pair in \(C_v\) are placed on the same two alignment columns. Hence each pair of strings contributes one unit to \(f_{\mathrm{str}}\), and
\[
f_{\mathrm{str}}(A')=\binom{k}{2}=\tau_{\mathrm{str}}.
\]
Thus the constructed \MSAS\ instance is feasible.

(\emph{Soundness})
Conversely, suppose the constructed \MSAS\ instance has a feasible alignment \(A'\). Since each string has exactly one designated pair, each pair of strings can contribute at most one unit to the pair-overlap score. Hence
\[
f_{\mathrm{str}}(A')\le \binom{k}{2}.
\]
Because feasibility requires
\[
f_{\mathrm{str}}(A')\ge \tau_{\mathrm{str}}=\binom{k}{2},
\]
equality must hold. Therefore, for every pair of strings, their unique designated pairs are placed on the same two alignment columns.

These designated pairs are precisely the two appended positions of each string.
Since alignments preserve the order of symbols within each row, all strings must place their first appended \(x\) in one common column and their second appended \(x\) in another common column, in the same order. 
Since these are the final two symbols in every row, no non-gap symbol can appear after the second common column. Because all-gap columns are excluded, these two common columns are the final two columns of the alignment. Neither column contains a gap.

Delete these two columns from \(A'\). The result is an alignment \(A\) of the original strings \(S_1,\ldots,S_k\). Deleting the final two all-\(x\) columns removes exactly the sequence-score contribution \(2\binom{k}{2}s(x,x)\), and does not affect the remaining gap-run penalties.

Therefore
\[
f_{\mathrm{seq}}(A)
=
f_{\mathrm{seq}}(A')-2\binom{k}{2}s(x,x)
\ge
\tau_{\mathrm{seq}}.
\]
Hence the original \textsf{Seq-MSA} instance is feasible.

The construction is polynomial-time computable. Since \(k\ge 2\) for the source instances, the constructed instance satisfies \(C_i\neq\varnothing\) for every \(i\) and \(\tau_{\mathrm{str}}=\binom{k}{2}>0\). Hence \MSAS\ remains NP-hard under this restriction.
\end{proof}

\section{No PTAS for the scalarized optimization problem (already for $k=2$)}\label{sec:apx}
This section proves that the scalarized optimization problem \(\MSAOPTlambda\) admits no polynomial-time approximation scheme (PTAS), already for two input strings. The proof proceeds by showing that a PTAS for \(\MSAOPTlambda\) would imply a PTAS for degree-\(1\) contact-map overlap (CMO), contradicting the known hardness of approximation for degree-\(1\) CMO.

\paragraph{Source problem: degree-\(1\) CMO.}
The source problem is contact-map overlap (CMO). In this section, CMO plays the role of a simplified abstraction of the pair-overlap score defined in Section~\ref{sec:model}. In the \textsf{MSA-S} model, each string has a set of position-pairs, and the pair-overlap score rewards such pairs when they are mapped to the same two alignment columns. CMO retains only the order of positions and which pairs of positions are selected.

More concretely, each CMO input can be viewed as an ordered graph: the vertices are linearly ordered positions, and the edges are selected pairs of positions. This edge structure is called a contact map. The objective of CMO is to find an order-preserving partial correspondence between two ordered vertex sets that preserves as many edges as possible.

Formally, for a positive integer \(n\), let
\[
[n]=\{1,2,\ldots,n\}
\]
denote a linearly ordered set of \(n\) positions, and let \(\binom{[n]}{2}\) denote the set of unordered pairs of distinct positions in \([n]\). A contact map on \([n]\), namely the edge structure on these ordered positions, is represented by a set
\[
D\subseteq \binom{[n]}{2}.
\]
Equivalently, \(D\) is the edge set of an undirected graph whose vertex set is the ordered position set \([n]\). An element \(\{i,i'\}\in D\) is an edge, meaning that the two positions \(i\) and \(i'\) of \([n]\) form a contact.

Given two contact maps
\[
D_1\subseteq \binom{[n]}{2},
\qquad
D_2\subseteq \binom{[m]}{2},
\]
a feasible CMO solution first selects subsets \(S\subseteq[n]\) and \(S'\subseteq[m]\) with \(|S|=|S'|\). It then matches each selected position \(i\in S\) to a selected position \(f(i)\in S'\), giving a bijection \(f:S\to S'\). This bijection must be order-preserving: for any \(i,i'\in S\), if \(i<i'\), then \(f(i)<f(i')\).

The value of this solution is the number of preserved edges, namely the number of edges \(\{i,i'\}\in D_1\) with \(i,i'\in S\) and \(\{f(i),f(i')\}\in D_2\).
We write
\[
\mathrm{CMO}(D_1,D_2)
\]
for the maximum value over all feasible CMO solutions.

We will use the restriction called degree-\(1\) CMO. In this restriction, each position is incident to at most one edge in each contact map. In other words, no position can participate in two different contacts. Equivalently, the undirected graph associated with each contact map has maximum degree \(1\). 
Thus each contact map is a matching.

This restricted problem is \(\MAXSNP\)-hard, indeed \(\MAXSNP\)-complete~\cite{Goldman1999}. Therefore, by the PCP-based hardness-of-approximation consequence for \(\MAXSNP\)-hard problems~\cite{Arora1998}, degree-\(1\) CMO admits no PTAS unless \(\Ptime=\NP\).

\paragraph{Target problem and scoring convention.}
The target problem of the reduction is the two-string restriction of \(\MSAOPTlambda\), defined in Section~\ref{sec:model}.
Throughout this section, we use the canonical unit non-gap score: identical non-gap symbols have score \(1\), and distinct non-gap symbols have score \(0\).
The affine gap parameters remain the fixed constants from Section~\ref{sec:model}, and the pair-overlap score \(f_{\mathrm{str}}\) is kept unchanged from that definition.

The scalarizing parameter \(\lambda\ge 1\) is a fixed rational constant and is not part of the input. For an alignment \(A\), define
\begin{equation}\label{eq:apx:val-opt}
\mathrm{val}_\lambda(A)
:=
f_{\mathrm{seq}}(A)+\lambda f_{\mathrm{str}}(A),
\qquad
\mathrm{OPT}_\lambda
:=
\max_A \mathrm{val}_\lambda(A),
\end{equation}
where \(A\) ranges over all alignments of the two input strings.

\paragraph{High-level reduction idea.}
The purpose of the reduction is to transfer approximation hardness from degree-\(1\) CMO to \(\MSAOPTlambda\). Suppose that \(\MSAOPTlambda\) admitted a PTAS. We construct an \(\MSAOPTlambda\) instance from a degree-\(1\) CMO instance so that a near-optimal alignment of the constructed instance can be converted into a near-optimal CMO solution. This would imply a PTAS for degree-\(1\) CMO, contradicting its known hardness.

The main idea is to make the pair-overlap score dominate the sequence score.
Indeed, the pair-overlap score \(f_{\mathrm{str}}\) has the same combinatorial form as CMO: both count pairs that are preserved under an order-preserving correspondence. To suppress the effect of symbol matching, we construct the two strings as \(a^L\) and \(b^L\) for distinct symbols \(a\neq b\). Under the canonical unit non-gap score, every non-gap symbol pair between the two strings has score \(0\). Thus the sequence score can only contribute through affine gap penalties.

To make one preserved contact contribute many pair-overlap terms, each original CMO position is represented by a block consisting of \(Q\) consecutive positions in the constructed strings. Thus, an original position is represented not by a single position, but by a block of size \(Q\).

If two CMO positions form a contact, then the reduced instance designates all pairs between the two corresponding blocks. In other words, one original contact is represented by the \(Q^2\) pairs between two blocks of size \(Q\). 
When this contact is preserved by a CMO solution, aligning the corresponding blocks in
the reduced instance makes these \(Q^2\) designated pairs overlap.
Hence one preserved contact contributes \(Q^2\) to the pair-overlap score of the constructed \(\MSAOPTlambda\) instance.

The block size \(Q\) is chosen large enough so that this \(Q^2\)-scaled pair-overlap contribution dominates the bounded total loss caused by affine gap penalties.
Finally, we show that a near-optimal alignment can be decoded back into an order-preserving CMO solution.

\subsection{Reduction construction}\label{sec:apx-construction}
We now give the formal construction corresponding to the high-level idea above. Starting from a degree-\(1\) CMO instance, we first pad the ordered position sets to a common size, then replace each position by a block of size \(Q\), and finally expand each contact into all position-pairs between the corresponding blocks.

Fix a degree-\(1\) CMO instance
\[
D_1\subseteq \binom{[n]}{2},
\qquad
D_2\subseteq \binom{[m]}{2}.
\]
Let
\[
N:=\max\{n,m\}.
\]
By adding isolated positions to the smaller ordered position set if necessary, we regard both \(D_1\) and \(D_2\) as contact maps on the common ordered position set \([N]\). No new contacts are added. This padding preserves degree-\(1\) and does not change the CMO optimum.

For the sequence-score bound below, define
\[
c_{\mathrm{gap}}:=|g_o|+|g_e|.
\]
This fixed constant bounds the magnitude of the loss caused by affine gap penalties. Since the fixed affine gap scheme is assumed to be nontrivial, we have \(c_{\mathrm{gap}}>0\).

Let \(\varepsilon\in(0,1)\) be the target approximation error for the source CMO instance. We will use the hypothetical PTAS for \(\MSAOPTlambda\) with accuracy parameter
\[
\delta:=\varepsilon/2.
\]
The block size \(Q\) is chosen so that the \(Q^2\)-amplified pair-overlap score dominates the bounded sequence-score loss caused by affine gap penalties. Set
\begin{equation}\label{eq:apx:block-size}
Q:=
\left\lceil
\frac{4c_{\mathrm{gap}}N}{\lambda\varepsilon}
\right\rceil,
\qquad
L:=NQ.
\end{equation}

To eliminate any positive contribution from non-gap symbol matching, choose distinct symbols \(a,b\in\Sigma\), and construct two strings
\[
S^{(1)}:=a^L,
\qquad
S^{(2)}:=b^L.
\]
Since \(a\neq b\), every non-gap symbol pair between \(S^{(1)}\) and \(S^{(2)}\) has score \(0\) under the canonical unit non-gap score. Nevertheless, gaps can still affect the alignment because they change the position-to-position correspondence induced by alignment columns, and hence can change the pair-overlap score. Thus the sequence-score term contributes only through affine gap penalties.

Within each constructed string, define the block assigned to CMO position \(i\) by
\[
B_i:=\{(i-1)Q+1,(i-1)Q+2,\ldots,iQ\}
\qquad (i\in[N]).
\]
Thus \(B_i\) is a block of \(Q\) consecutive string positions.

\paragraph{Expanded designated position-pair sets.}
We now specify how the CMO contacts are represented as designated position-pairs in the reduced instance.

For each \(\nu\in\{1,2\}\), let
\[
\widehat C_\nu\subseteq \{(p,q):1\le p<q\le L\}
\]
be the designated position-pair set for the constructed string \(S^{(\nu)}\).
For every contact \(\{i,i'\}\in D_\nu\) with \(i<i'\), we include in
\(\widehat C_\nu\) every pair
\[
(p,q)
\quad\text{with}\quad
p\in B_i,\ q\in B_{i'}.
\]
Thus the endpoint \(i\) is replaced by a position \(p\) in \(B_i\), and the endpoint \(i'\) is replaced by a position \(q\) in \(B_{i'}\). Since both blocks have size \(Q\), one CMO contact gives \(Q^2\) designated position-pairs.
Since \(i<i'\), every such pair satisfies \(p<q\). No other designated position-pairs are added.

The reduced \(\MSAOPTlambda\) instance is therefore
\[
(S^{(1)},S^{(2)};\ \widehat C_1,\widehat C_2).
\]

\subsection{Basic bounds and embedding}
We next establish two basic properties of the reduced instance: the sequence score is uniformly bounded, and any CMO solution of value \(t\) can be embedded as an alignment with pair-overlap score \(Q^2t\).

\begin{lemma}[Uniform bound on the \(f_{\mathrm{seq}}\) term]
\label{lem:apx:fseq-bound}
For any alignment \(A\) of \((S^{(1)},S^{(2)})\),
\[
-2c_{\mathrm{gap}}L \le f_{\mathrm{seq}}(A)\le 0.
\]
\end{lemma}

\begin{proof}
Since \(a\neq b\), every non-gap symbol pair between \(S^{(1)}\) and \(S^{(2)}\) has score \(s(a,b)=0\) under the canonical unit non-gap score.
Thus the only possible contributions to \(f_{\mathrm{seq}}(A)\) come from affine gap penalties.

Because \(g_o\le 0\) and \(g_e\le 0\), every gap-run contribution
\(g_o+g_e\ell_r\) is non-positive. Hence
\[
f_{\mathrm{seq}}(A)\le 0.
\]

For the lower bound, let \(r\) range over the gap runs in the two rows of \(A\), and let \(\ell_r\ge 1\) be the length of run \(r\). By the definition \(c_{\mathrm{gap}}=|g_o|+|g_e|\), each gap run satisfies
\[
g_o+g_e\ell_r
=
-|g_o|-|g_e|\ell_r
\ge
-(|g_o|+|g_e|)\ell_r
=
-c_{\mathrm{gap}}\ell_r,
\]
where we use \(\ell_r\ge 1\). Therefore
\[
f_{\mathrm{seq}}(A)
\ge
-c_{\mathrm{gap}}\sum_r \ell_r.
\]
In an alignment of two strings of length \(L\), the total number of gap symbols is at most \(2L\). Hence \(\sum_r\ell_r\le 2L\), and therefore
\[
f_{\mathrm{seq}}(A)\ge -2c_{\mathrm{gap}}L.
\]
\end{proof}

\begin{lemma}[Embedding a CMO solution]\label{lem:apx:embed}
Let \(t\) be the value of a feasible CMO solution for \((D_1,D_2)\).
Then there exists an alignment \(A\) of the reduced instance such that
\[
f_{\mathrm{str}}(A)=Q^2 t
\qquad\text{and}\qquad
\mathrm{val}_\lambda(A)\ge \lambda Q^2 t - 2c_{\mathrm{gap}}L.
\]
\end{lemma}

\begin{proof}
Let \(f:S\to S'\) be an order-preserving bijection whose CMO value is \(t\).

We first construct an alignment \(A\) from this CMO solution. For each \(i\in S\), align the block \(B_i\) in \(S^{(1)}\) positionwise with the block \(B_{f(i)}\) in \(S^{(2)}\). The remaining blocks, namely blocks \(B_i\) in \(S^{(1)}\) with \(i\notin S\) and blocks \(B_j\) in \(S^{(2)}\) with \(j\notin S'\), are aligned against gaps. Since \(f\) is order-preserving, all these block alignments can be placed in order without crossings.

We next compute the pair-overlap score of this alignment. Consider a contact
\(\{i,i'\}\in D_1\) preserved by \(f\), so
\[
\{f(i),f(i')\}\in D_2.
\]
The contact \(\{i,i'\}\) is expanded into all \(Q^2\) designated pairs between \(B_i\) and \(B_{i'}\). Similarly, the contact \(\{f(i),f(i')\}\) is expanded into all \(Q^2\) designated pairs between \(B_{f(i)}\) and \(B_{f(i')}\).
Since the corresponding blocks are aligned positionwise, these \(Q^2\) pairs occupy the same pairs of alignment columns. Thus each preserved CMO contact contributes exactly \(Q^2\) to \(f_{\mathrm{str}}\).

No other contact of \(D_1\) contributes to \(f_{\mathrm{str}}\) in this constructed alignment. If at least one of its endpoint blocks is not selected by the CMO solution, then positions in that block are aligned only against gaps. If both endpoint blocks are selected but their images do not form a contact in \(D_2\), then there is no corresponding designated pair in \(\widehat C_2\).
Hence the only pair-overlap contributions come from the \(t\) contacts preserved by \(f\), and therefore
\[
f_{\mathrm{str}}(A)=Q^2t.
\]

Finally, by Lemma~\ref{lem:apx:fseq-bound},
\[
f_{\mathrm{seq}}(A)\ge -2c_{\mathrm{gap}}L.
\]
Hence
\[
\mathrm{val}_\lambda(A)
=
f_{\mathrm{seq}}(A)+\lambda f_{\mathrm{str}}(A)
\ge
-2c_{\mathrm{gap}}L+\lambda Q^2t.
\]
\end{proof}

As an immediate consequence, let
\[
t^\star:=\mathrm{CMO}(D_1,D_2)
\]
be the optimum value of the original CMO instance. Applying Lemma~\ref{lem:apx:embed} to an optimal CMO solution gives
\begin{equation}\label{eq:apx:opt-lower}
\mathrm{OPT}_\lambda
\ge
\lambda Q^2 t^\star - 2c_{\mathrm{gap}}L.
\end{equation}

We also note that the reduced instance has nonnegative optimum. Indeed, the gap-free alignment of \(S^{(1)}\) and \(S^{(2)}\) has \(f_{\mathrm{seq}}=0\), because every non-gap pair is \(a\) aligned with \(b\), and \(s(a,b)=0\). Since \(f_{\mathrm{str}}\ge 0\), this alignment has nonnegative value. Hence
\[
\mathrm{OPT}_\lambda\ge 0.
\]

If \(t^\star=0\), then the CMO instance is trivial for approximation purposes. For the nontrivial case \(t^\star\ge 1\), the choice of \(Q\) gives
\[
2c_{\mathrm{gap}}L
=
2c_{\mathrm{gap}}NQ
\le
\frac{\lambda\varepsilon}{2}Q^2
\le
\frac{\lambda\varepsilon}{2}Q^2t^\star.
\]
Combining this with \eqref{eq:apx:opt-lower}, we obtain
\begin{equation}\label{eq:apx:opt-lower-scaled}
\mathrm{OPT}_\lambda
\ge
\left(1-\frac{\varepsilon}{2}\right)\lambda Q^2t^\star.
\end{equation}

\subsection{Monotonicity and exact decomposition of the pair-overlap term}
We now analyze an arbitrary alignment \(A\) of the reduced instance. The goal is to express the pair-overlap score \(f_{\mathrm{str}}(A)\) in terms of correspondences between contacts of \(D_1\) and contacts of \(D_2\). This is the first step in the reverse direction: later, we will use this decomposition to extract a CMO solution from a near-optimal alignment.

We first define block-level correspondence counts. For \(i,j\in[N]\), let
\[
x_{ij}:=
\#\{(p,p'):\ p\in B_i,\ p'\in B_j,\ (1,p)\leftrightarrow_A (2,p')\}.
\]
Thus \(x_{ij}\) counts how many positions in block \(B_i\) of \(S^{(1)}\) are aligned with positions in block \(B_j\) of \(S^{(2)}\).

Since each block contains \(Q\) positions and each position is aligned to at most one position in the other string,
\[
\sum_{j=1}^N x_{ij}\le Q,
\qquad
\sum_{i=1}^N x_{ij}\le Q.
\]

The order-preserving property of alignments implies the following monotonicity.
If \(i<i'\) and \(j<j'\), then
\[
x_{ij'}x_{i'j}=0.
\]
Indeed, if both factors were positive, then there would be aligned pairs
\((p,p')\) and \((q,q')\) with
\[
p\in B_i,\quad p'\in B_{j'},\qquad
q\in B_{i'},\quad q'\in B_j.
\]
Since \(i<i'\), every position in \(B_i\) precedes every position in \(B_{i'}\), so \(p<q\). 
Since \(j<j'\), every position in \(B_j\) precedes every position in \(B_{j'}\), so \(q'<p'\). 
This contradicts the order-preserving property of alignments.

For contacts \(e=\{i,i'\}\in D_1\) and \(e'=\{j,j'\}\in D_2\), written with \(i<i'\) and \(j<j'\), define
\begin{equation}\label{eq:apx:omega}
\Omega_{e,e'}(A):=x_{ij}x_{i'j'}.
\end{equation}

\begin{lemma}[Exact decomposition of \(f_{\mathrm{str}}\)]
\label{lem:apx:fstr-decomposition}
For any alignment \(A\) of the reduced instance,
\[
f_{\mathrm{str}}(A)
=
\sum_{e\in D_1}\sum_{e'\in D_2}\Omega_{e,e'}(A).
\]
\end{lemma}

\begin{proof}
Fix contacts \(e=\{i,i'\}\in D_1\) and \(e'=\{j,j'\}\in D_2\), written with \(i<i'\) and \(j<j'\). By construction, \(e\) is expanded into all designated pairs between \(B_i\) and \(B_{i'}\) in \(\widehat C_1\), and \(e'\) is expanded into all designated pairs between \(B_j\) and \(B_{j'}\) in \(\widehat C_2\).

We count the contribution to \(f_{\mathrm{str}}\) coming from this fixed pair of contacts. A pair-overlap between these two expanded contacts is obtained by choosing one aligned position pair between \(B_i\) and \(B_j\), and one aligned position pair between \(B_{i'}\) and \(B_{j'}\). The number of such choices is
\[
x_{ij}x_{i'j'}.
\]
The reversed correspondence, namely \(B_i\) to \(B_{j'}\) and \(B_{i'}\) to \(B_j\), would contribute \(x_{ij'}x_{i'j}\), but this term is zero by the monotonicity argument above. Hence the contribution of the contact pair \((e,e')\) is exactly
\[
x_{ij}x_{i'j'}=\Omega_{e,e'}(A).
\]

Summing over all contacts \(e\in D_1\) and \(e'\in D_2\) gives the result.
\end{proof}

\subsection{Fractional matching extraction}
We now reinterpret the decomposition above as a weighted bipartite graph on contacts. The left vertices are the contacts of \(D_1\), and the right vertices are the contacts of \(D_2\). An alignment \(A\) induces a weight between two contacts whenever it creates a positive pair-overlap contribution between them.

The reason for using fractional weights is that an arbitrary alignment need not match one contact of \(D_1\) integrally to one contact of \(D_2\). Since each CMO position has been expanded into a block of \(Q\) positions, a block may be split across several blocks on the other side. Consequently, one contact may contribute partially to several contacts. The weights below record these partial contact-level correspondences.

Formally, for an alignment \(A\), define the auxiliary bipartite graph
\[
H_A=(D_1,D_2,E_A)
\]
by
\[
(e,e')\in E_A
\quad\Longleftrightarrow\quad
\Omega_{e,e'}(A)>0.
\]
For \(e\in D_1\) and \(e'\in D_2\), define
\[
w_{e,e'}(A):=\frac{\Omega_{e,e'}(A)}{Q^2}.
\]
Thus \(w_{e,e'}(A)=0\) whenever \((e,e')\notin E_A\).

The normalization by \(Q^2\) is chosen because each contact is expanded into exactly \(Q^2\) designated position-pairs, while an arbitrary alignment can realize at most \(Q^2\) total pair-overlap contribution incident to a single contact. Hence, after division by \(Q^2\), the total weight incident to any single contact is at most \(1\). In the next lemma, we show that these normalized weights form a feasible fractional matching on \(H_A\). Then, using integrality of bipartite matching, this fractional matching will be converted into an ordinary matching of comparable size.

\begin{lemma}[The weights form a feasible fractional matching]
\label{lem:apx:fractional-matching}
For any alignment \(A\), the weights
\[
w_{e,e'}(A):=\frac{\Omega_{e,e'}(A)}{Q^2}
\]
form a feasible fractional matching on the auxiliary bipartite graph \(H_A\).
That is,
\[
\sum_{e'\in D_2} w_{e,e'}(A)\le 1
\quad\text{for all }e\in D_1,
\qquad
\sum_{e\in D_1} w_{e,e'}(A)\le 1
\quad\text{for all }e'\in D_2.
\]
Moreover,
\[
\sum_{e\in D_1}\sum_{e'\in D_2} w_{e,e'}(A)
=
\frac{f_{\mathrm{str}}(A)}{Q^2}.
\]
\end{lemma}

\begin{proof}
Fix a contact \(e=\{i,i'\}\in D_1\), written with \(i<i'\).
The total contribution incident to this fixed contact is bounded by allowing its two endpoint blocks to connect to arbitrary blocks on the other side. Thus
\[
\sum_{e'\in D_2} \Omega_{e,e'}(A)
\le
\sum_{j=1}^{N}\sum_{j'=1}^{N} x_{ij}x_{i'j'}
=
\Bigl(\sum_{j=1}^N x_{ij}\Bigr)
\Bigl(\sum_{j'=1}^N x_{i'j'}\Bigr)
\le Q^2.
\]
Dividing by \(Q^2\) gives
\[
\sum_{e'\in D_2} w_{e,e'}(A)\le 1.
\]

The same argument with the two sides reversed gives the corresponding bound for a fixed contact \(e'=\{j,j'\}\in D_2\). Indeed,
\[
\sum_{e\in D_1} \Omega_{e,e'}(A)
\le
\sum_{i=1}^{N}\sum_{i'=1}^{N} x_{ij}x_{i'j'}
=
\Bigl(\sum_{i=1}^N x_{ij}\Bigr)
\Bigl(\sum_{i'=1}^N x_{i'j'}\Bigr)
\le Q^2.
\]
Hence
\[
\sum_{e\in D_1} w_{e,e'}(A)\le 1.
\]

Finally, by Lemma~\ref{lem:apx:fstr-decomposition},
\[
\sum_{e\in D_1}\sum_{e'\in D_2}w_{e,e'}(A)
=
\frac{1}{Q^2}
\sum_{e\in D_1}\sum_{e'\in D_2}\Omega_{e,e'}(A)
=
\frac{f_{\mathrm{str}}(A)}{Q^2}.
\]
\end{proof}

The next lemma converts the fractional contact-level weights into an ordinary matching in the auxiliary bipartite graph.

\begin{lemma}[Integral matching extraction]\label{lem:apx:matching-extraction} From any alignment \(A\), one can compute in polynomial time a matching \(M\) in \(H_A\) such that
\[
|M|\ge \frac{f_{\mathrm{str}}(A)}{Q^2}.
\]
\end{lemma}

\begin{proof}
By Lemma~\ref{lem:apx:fractional-matching}, the weights \(w_{e,e'}(A)\) define a feasible fractional matching on the bipartite graph \(H_A\) with value
\[
\sum_{e\in D_1}\sum_{e'\in D_2} w_{e,e'}(A)
=
\frac{f_{\mathrm{str}}(A)}{Q^2}.
\]
Since the bipartite matching polytope is integral, the maximum integral matching size in \(H_A\) is at least the value of this feasible fractional matching.
Equivalently, in a bipartite graph, the maximum fractional matching value equals the maximum integral matching size.
Therefore there exists a matching \(M\subseteq E_A\) with
\[
|M|\ge \frac{f_{\mathrm{str}}(A)}{Q^2}.
\]
Such a matching can be computed in polynomial time by a standard bipartite maximum-matching algorithm.
\end{proof}

The matching \(M\) obtained here matches contacts of \(D_1\) to contacts of \(D_2\). In the next step, we convert this contact-level matching into a CMO feasible solution.

\subsection{From a contact matching to a feasible CMO solution}
We now convert a contact-level matching \(M\subseteq E_A\) in the auxiliary graph \(H_A\) into a feasible CMO solution. This is where the degree-\(1\) assumption is used: because contacts in each contact map are vertex-disjoint, a matching between contacts determines a well-defined partial correspondence between their endpoints.

\begin{lemma}[A matching in \(H_A\) yields an order-preserving partial bijection]
\label{lem:apx:matching-to-cmo}
Let \(M\) be any matching in \(H_A\). Then one can compute in polynomial time a feasible CMO solution of value at least \(|M|\).
\end{lemma}

\begin{proof}
For each matched pair \((e,e')\in M\), write
\[
e=\{i,i'\}\in D_1,\qquad e'=\{j,j'\}\in D_2
\]
with \(i<i'\) and \(j<j'\). Define
\[
f(i):=j,\qquad f(i'):=j'.
\]
This defines \(f\) on the endpoints of contacts of \(D_1\) that appear in matched pairs of \(M\).

This assignment turns the contact-level matching into a partial map on positions. We verify that it is well-defined, injective, and order-preserving.

Since \(M\) is a matching in \(H_A\), no contact of \(D_1\) or \(D_2\) is used more than once. Moreover, \(D_1\) and \(D_2\) are degree-\(1\) contact maps, so their contacts are vertex-disjoint. Hence no position receives two conflicting images, and the resulting map is injective on its domain.

It remains to show that \(f\) is order-preserving. Let \(u<v\) be two positions in the domain of \(f\). For any position \(z\) in the domain of \(f\), the matched pair defining \(f(z)\) is an edge of \(H_A\). 
Since \(\Omega_{e,e'}(A)>0\), both endpoint correspondences used by this matched contact pair have positive block-level count. In particular,
\[
x_{z,f(z)}>0.
\]
Applying this to \(z=u\) and \(z=v\), there exist aligned position pairs
\[
(p,p')\in B_u\times B_{f(u)},
\qquad
(q,q')\in B_v\times B_{f(v)}.
\]
Since \(u<v\), every position in \(B_u\) precedes every position in \(B_v\), so \(p<q\). By the order-preserving property of alignments, this implies \(p'\le q'\). Since the blocks are ordered consecutively, \(p'\le q'\) implies \(f(u)\le f(v)\). Since \(f\) is injective, we conclude
\[
f(u)<f(v).
\]
Therefore \(f\) is order-preserving.

Let \(S\) be the domain of \(f\) and let \(S':=f(S)\). 
Then \((S,S',f)\) is a feasible CMO solution. Each matched pair \((e,e')\in M\) maps one contact of \(D_1\) to a contact of \(D_2\), so the value of this CMO solution is at least \(|M|\). The construction is polynomial-time.
\end{proof}

\subsection{Approximation transfer}
We now complete the reverse direction of the reduction. Using the lower bound on \(\mathrm{OPT}_\lambda\) obtained from the embedding lemma and the extraction lemmas above, we show that a near-optimal alignment of the reduced instance yields a near-optimal solution to the original degree-\(1\) CMO instance.

\begin{lemma}[Near-optimal alignment yields near-optimal CMO solution]
\label{lem:apx:extract}
Let \(t^\star:=\mathrm{CMO}(D_1,D_2)\), and let \(A\) be an alignment of the reduced instance such that
\[
\mathrm{val}_\lambda(A)\ge (1-\delta)\mathrm{OPT}_\lambda,
\qquad
\delta=\varepsilon/2.
\]
Then one can compute in polynomial time a feasible CMO solution of value at least \((1-\varepsilon)t^\star\).
\end{lemma}

\begin{proof}
If \(t^\star=0\), the empty CMO solution already has value \(0=(1-\varepsilon)t^\star\), so the claim is trivial. Hence assume \(t^\star\ge 1\).

By Lemma~\ref{lem:apx:fseq-bound}, \(f_{\mathrm{seq}}(A)\le 0\). Therefore
\[
\lambda f_{\mathrm{str}}(A)
\ge
\mathrm{val}_\lambda(A).
\]
Using the near-optimality of \(A\) and the lower bound \eqref{eq:apx:opt-lower-scaled}, we obtain
\[
\lambda f_{\mathrm{str}}(A)
\ge
\mathrm{val}_\lambda(A)
\ge
(1-\delta)\mathrm{OPT}_\lambda
\ge
\left(1-\frac{\varepsilon}{2}\right)^2
\lambda Q^2t^\star,
\]
where \(\delta=\varepsilon/2\). Since
\[
\left(1-\frac{\varepsilon}{2}\right)^2
\ge
1-\varepsilon,
\]
it follows that
\[
\frac{f_{\mathrm{str}}(A)}{Q^2}
\ge
(1-\varepsilon)t^\star.
\]

By Lemma~\ref{lem:apx:matching-extraction}, we can compute a matching
\(M\) in \(H_A\) with
\[
|M|\ge \frac{f_{\mathrm{str}}(A)}{Q^2}.
\]
Then, by Lemma~\ref{lem:apx:matching-to-cmo}, this matching yields a feasible CMO solution of value at least \(|M|\). Therefore the extracted CMO solution has value at least \((1-\varepsilon)t^\star\).
\end{proof}

\begin{theorem}[No PTAS for \(\MSAOPTlambda\)] 
\label{thm:apx:no-ptas}
For any fixed rational constant \(\lambda\ge 1\), under the canonical unit non-gap scoring scheme, \(\MSAOPTlambda\) admits no PTAS unless \(\Ptime=\NP\), already for \(k=2\).
\end{theorem}

\begin{proof}
Assume that, for the fixed \(\lambda\ge 1\), \(\MSAOPTlambda\) admits a PTAS.
Given a degree-\(1\) CMO instance and an accuracy parameter \(\varepsilon\in(0,1)\), construct the reduced instance above and apply the PTAS with accuracy parameter \(\delta=\varepsilon/2\). The reduced instance is constructed in polynomial time for fixed \(\varepsilon\) and \(\lambda\).

The PTAS returns an alignment \(A\) satisfying
\[
\mathrm{val}_\lambda(A)\ge (1-\delta)\mathrm{OPT}_\lambda.
\]
By Lemma~\ref{lem:apx:extract}, we can then compute in polynomial time a feasible CMO solution of value at least
\[
(1-\varepsilon)\mathrm{CMO}(D_1,D_2).
\]
Thus a PTAS for \(\MSAOPTlambda\) would imply a PTAS for degree-\(1\) CMO.
This contradicts the no-PTAS consequence of the \(\MAXSNP\)-hardness of degree-\(1\) CMO~\cite{Goldman1999,Arora1998}, unless \(\Ptime=\NP\).
\end{proof}

\section{Conclusion}
We introduced a formal structure-informed multiple sequence alignment problem, \textsf{MSA-S}, in which a fixed pairwise string score is combined with a binary overlap score on designated position-pairs. For a broad class of fixed pairwise string scoring schemes, we established NP-completeness of the associated decision problem, and further showed that NP-hardness persists under a nonempty designated-pair restriction. We also proved that, under the canonical unit scheme for the non-gap symbol-pair scoring rule, the associated scalarized optimization problem admits no PTAS even for two input strings \((k=2)\), unless \(\mathsf{P}=\mathsf{NP}\).
These results provide a formal complexity-theoretic baseline for subsequent study of structure-informed multiple sequence alignment, showing that even a minimal fixed-score formulation is difficult both in its decision form and in its scalarized approximation form.

\section{Acknowledgement}
Kanazawa Yoshiki acknowledges support from the Taikichiro Mori Memorial Research Grants and the Kenkyu-no-Susume Scholarship.

\bibliographystyle{unsrt}
\bibliography{ref}

\end{document}